\documentclass[10pt]{article}
\usepackage{amstext, amsmath, amsthm, amssymb}
\usepackage{amsmath}
\usepackage{geometry} 
\usepackage{color}

\numberwithin{equation}{section}
\usepackage[nottoc]{tocbibind}
\usepackage{hyperref}
\usepackage{makeidx}
\makeindex

\newtheorem{thm}{Theorem}[section]

\newtheorem{lem}[thm]{Lemma}

\newtheorem{defin}[thm]{Definition}
\newtheorem{rem}[thm]{Remark}

\usepackage{colortbl}

\newcommand{\C}{\mathbb{C}}
\newcommand{\R}{\mathbb{R}}

\title{A Poisson Algebra on the Hida Test Functions and a Quantization using the Cuntz Algebra
}

\author{Wolfgang Bock, Vyacheslav Futorny, Mikhail Neklyudov}

\AtEndDocument{\bigskip{\footnotesize%
  (W.\,Bock) \textsc{Technomathematics Group,
University of Kaiserslautern,
P. O. Box 3049, 67653 Kaiserslautern, Germany} \par
  \textit{E-mail address}: \texttt{bock@mathematik.uni-kl.de} \par
  \addvspace{\medskipamount}
 (V.\,Futorny) \textsc{Instituto de Matematica e Estatistica, Universidade de Sa\~{o} Paulo, Caixa Postal 66281,\\ Sa\~{o} Paulo, CEP 05315-970, Brasil} \par
  \textit{E-mail address}: \texttt{vfutorny@gmail.com}  \par
   \addvspace{\medskipamount}
  (M.\,Neklyudov) \textsc{Instituto de Ci\^{e}ncias Exatas, Departamento de Matematica, UFAM, Manaus, CEP 69077-000, Brasil} \par
  \textit{E-mail address}: \texttt{misha.neklyudov@gmail.com}
}}

\date{}

\begin{document}
\maketitle

\section{Introduction}

In this note we define one more way of quantization (see review \cite{AE2005} and references therein) of classical systems. 
The quantization we consider is an analogue of classical Jordan-Schwinger (J.-S.) map which has been known and used for a long time by physicists (\cite{BidernharnLouck1981}). The difference, comparing to J.-S. map, is that we use generators of Cuntz algebra $\mathcal{O}_{\infty}$ (i.e. countable family of mutually orthogonal partial isometries of separable Hilbert space) as a "building blocks" instead of creation-annihilation operators. The resulting scheme satisfies properties similar to Van Hove prequantization i.e. exact conservation of Lie bracket and linearity. Second result of the paper is a construction of representation of Heisenberg algebra through Cuntz generators (remark \ref{CCRconstruction}). Other way of construction  of CCR relations through isometries has been presented in paper \cite{Kaw2007}. The difference is that our construction is through an explicit formula while they construct the operators through recursive process. Furthermore, their iterative process results in polynomials of Cuntz generators of arbitrarily high degree while in our case we have quadratic dependence upon Cuntz generators.

Theory of representations of the algebra $\mathcal{O}_{\infty}$ (see \cite{DutJoerg2014} and references therein) seems to be much richer then the theory of representations of CCR relations. In particular, there is no analogue of Stone-Von Neumann theorem and classification of classes of irreducible representations is connected with completely different areas such as the theory of modular classes (\cite{KawHayLas2009}) and wavelet theory (\cite{BratelliJoergensen99})! Furthermore, as shown in \cite{Glimm1960} classification of all irreducible representations is in certain sense impossible.
The hope of the authors is that this variety of ideas could be connected to the quantization theory and result in new insights.


\section{Quantization of finite systems via Cuntz algebras}

Let $(C^{\infty}(M),\{\cdot,\cdot\})$ be a Poisson manifold with Poisson bracket $\{\cdot,\cdot\}:C^{\infty}(M)\times C^{\infty}(M)\to C^{\infty}(M)$, $H$ an auxiliary separable Hilbert space and 
$\{S_i\}_{i=1}^{\infty}:H\to H$ generators of Cuntz algebra $\mathcal{O}_{\infty}$ i.e. mutually orthogonal isometries of $H$ (\cite{Cuntz1977}, \cite{BratelliJoergensen99},\cite{Dutkay2014}). We can assume without loss of generality that
\[
\sum\limits_{k=1}^{\infty} S_k S_k^*=Id.
\]
Define operators $Q,R\in\mathcal{L}(C^{\infty}(M),\mathcal{O}_{\infty})$ as follows
\[
Q(h):=\sum\limits_{i,j=1}^{\infty}<\{h,e_j\},f_i>_{C^{\infty}(M),(C^{\infty}(M))^*} S_i  S_j^*,
\]
\[
R(h):=\sum\limits_{i,j=1}^{\infty}<he_j,f_i>_{C^{\infty}(M),(C^{\infty}(M))^*} S_i  S_j^*, h\in C^{\infty}(M),
\]
where $\{e_i\}_{i=1}^{\infty},\{f_i\}_{i=1}^{\infty}$ is a biorthogonal system in $C^{\infty}(M)$ (with some fixed dual $(C^{\infty}(M))^*$).
Then we have
\begin{lem}\label{lem:ComRel}
\begin{equation}
[Q(f),Q(g)]=Q(\{f,g\}),\quad[Q(f),R(g)]=R(\{f,g\}),\label{eqn:Quantizationf_1}
\end{equation}
\begin{equation}
Q(g)R(f)+Q(f)R(g)=Q(fg),\quad R(f)R(g)=R(fg), f,g\in C^{\infty}(M)\label{eqn:Leibnitzf_1}
\end{equation}
\end{lem}

\begin{proof}
It immediately follows from commutation properties of operators $\{S_i,S_j^*\}_{i,j=1}^{\infty}$ and Poisson bracket properties.
\end{proof}
\begin{defin}
Define quantization $\widehat{Q}\in \mathcal{L}(C^{\infty}(M),\mathcal{O}_{\infty})$ as follows
\[
\widehat{Q}:=R-2i Q.
\]
\end{defin}
\begin{thm}\label{thm:QuantFin}
$\widehat{Q}\in \mathcal{L}(C^{\infty}(M),\mathcal{O}_{\infty})$ satisfies
\begin{eqnarray}
\widehat{Q}(1) &= Id,\label{eqn:IdProp}\\ 
{}[\widehat{Q}(f),\widehat{Q}(g)] &= -2i \widehat{Q}(\{f,g\}),\label{eqn:LieBracketProp}\\ 
{}[\widehat{Q}(q_k),\widehat{Q}(q_j)] &= [\widehat{Q}(p_k),\widehat{Q}(p_j)]=0, [\widehat{Q}(q_k),
\widehat{Q}(p_j)] =  -2i \delta_{kj}Id, k,j=1,\ldots, \mathrm{dim} M\label{eqn:CCRAnalogueProp}
\end{eqnarray}
Furthermore, if $\phi:\mathbb{R}\to\mathbb{R}$ is an analytic function then
\begin{equation}
\Re \widehat{Q}(\phi(f))=\phi(\Re\widehat{Q}(f))\\\label{eqn:vonNeumannRuleProp}
\end{equation}
\end{thm}
\begin{proof}
Property \eqref{eqn:IdProp} follows from definition of $Q$ and $R$, commutation relation \eqref{eqn:LieBracketProp} is a consequence of Lemma \ref{lem:ComRel}, property \eqref{eqn:CCRAnalogueProp} immediately follows from \eqref{eqn:LieBracketProp} and Poisson bracket properties. At last, analogue of von Neumann rule is enough  to show when $\phi(x)=x^n$ is a monomial. Now the result follows by induction w.r.t. $n$ (applying properties\eqref{eqn:Leibnitzf_1}).
\end{proof}
\begin{rem}\label{CCRconstruction}
Mapping $Q$ itself satisfies property \eqref{eqn:LieBracketProp}, but we have that $Q(1)=0$. Nevertheless, working separately with $Q$ and $P$ allows us to get representation of Canonical Commutation Relations (CCR) as following example shows.
%

Let $M=\mathbb{R}^{2n}$ with the standard Poisson bracket, $\{e_i\}_{i=1}^{\infty}$ be an orthonormal basis in $L^2(\mathbb{R}^{2n},d\mu)$, $f_i=e_i,i\in \mathbb{N}$, $\mu$--standard Gaussian measure($d\mu=e^{-\sum\limits_{i=1}^n p_i^2+q_i^2}\prod\limits_i dp_idq_i$), and as the duality we take scalar product in $L^2(\mathbb{R}^{2n},d\mu)$. Then, as in the previous example,
\[
Q_{2i}:=Q(q_i),Q_{2i+1}:=Q(p_i), P_{2i}:=R(p_i), P_{2i+1}:=R(q_i), i=1,\ldots,2n,
\]
\[
[Q_i,Q_j]=[P_i,P_j]=0, [Q_i,P_j]=(-1)^i\delta_{ij}1,i,j=1,\ldots,2n.
\]
Furthermore, by integration by parts, we can deduce that
\[
P_i=(-1)^i(Q_i+Q_i^*),i=1,\ldots,2n.
\]
Therefore, we can conclude that
\[
[Q_i,Q_j]=[Q_i^*,Q_j^*]=0, [Q_i,Q_j^*]=\delta_{ij}1,i,j=1,\ldots,2n,
\]
and formula \eqref{eqn:Leibnitzf_1} allows us to calculate $Q(f),R(f)$ for arbitrary polynomial $f=f(q,p)$ as a polynomial of operators $Q_i,Q_i^*,i=1,\ldots,2n$.
\end{rem}
\begin{rem}
Notice that operators $P_k=S_kS_k^*,k\in \mathbb{N}$ are mutually orthogonal projections. Consequently, we have representation of $H$ as a direct sum
\[
H=\oplus_{k=1}^{\infty} H_k, H_k:=P_k(H).
\]
Let us show that operators $Q(h), R(h)$ are bounded on each $H_k, k\in\mathbb{N}$ under some natural assumptions about $h$. We will consider only the operator $Q(h)$. The case of $R(h)$ is similar. First, let us notice that
\[
Q(h)S_kS_k^*\psi=\sum\limits_{i=1}^{\infty}<\{h,e_k\},f_i>S_iS_k^*\psi.
\]
Consequently, by mutual orthogonality of isometries $\{S_l\}_{l=1}^{\infty}$ we can deduce that
\[
||Q(h)S_kS_k^*\psi||_H^2=||S_k^*\psi||_H^2\sum\limits_{l=1}^{\infty}<\{h,e_k\},f_l>^2=||S_kS_k^*\psi||_H^2\sum\limits_{l=1}^{\infty}<\{h,e_k\},f_l>^2,
\]
and, therefore, 
\[
||Q(h)||_{H_k}^2\leq \sum\limits_{l=1}^{\infty}<\{h,e_k\},f_l>^2.
\]
Thus if we assume that for any $k\in\mathbb{N}$ $ \sum\limits_{l=1}^{\infty}<\{h,e_k\},f_l>^2<\infty$ we have that $Q(h)$ has dense in $H$  domain of definition $\mathcal{D}=\{\mbox{finite linear combinations of elements of the subspaces $H_k$, $k\in\mathbb{N}$}\}$.

\end{rem}

\section{An Infinite Dimensional Extension via White Noise Calculus}
Starting point of the white noise distribution theory is the Gel'fand triple 
$$ S \subset L^2(\R,dt) \subset S^*,$$
where $S$ is the space of Schwartz test functions over $\R$ densely embedded in the Hilbert space of square integrable functions with respect to the  Lebesgue measure $L^2(\R,dt)$ and $S^*$ the space of tempered distributions, see. e.g.~\cite{Schaefer1} for a construction.\\
Via the Bochner-Minlos-Sazonov theorem, see e.g.~\cite{BK95}, we obtain the white noise measure $\mu$ on $S^*$ by its Fourier transform
$$\int_{S^*} \exp(i\langle x , \xi \rangle ) \, d\mu(x) = \exp(-\frac{1}{2} |\xi |_0^2), \quad \xi \in S, 
$$
where $|.|_0$ denotes the Hilbertian norm on $L^2(\R,dt)$. 
The topology on $S$ is induced by a positive self-adjoint operator $A$ on the space of real-valued functions $H:= L^2(\R,dt)$ with $\inf \sigma (A) > 1$ and Hilbert-Schmidt inverse $A^{-1}$ . We set $\rho := \left\|A ^{-1}\right\|_{OP}$ and $\delta := \left\|A ^{-1}\right\|_{HS}$. Note that the complexification $S_{\C}$ are equipped with the norms $\left|\xi\right|_p$ := $\left|A^p\xi\right|_0$ for $p \in \R$. We denote $H_\C :=  L^2(\R,\C,dt)$ furthermore  $S_{\C,p} := \left\{\xi\in S_\C|\ \left|\xi\right|_p <\infty\right\}$ and  $S^*_p := \left\{\xi\in S^*|\ \left|\xi\right|_p <\infty\right\}$, for $p \in \R$ resp.\\
Now we consider the following Gel'fand triple of Hida test functions and Hida distributions.
\begin{equation*}
	(S)_\beta \subset (L^2) := L^2(S^*,\mu) \subset (S)^*_\beta,\quad 0\leq \beta < 1 
\end{equation*} 
By the Wiener-Ito chaos decomposition theorem, see e.g.~\cite{Hida1, Obata1, Kuo2} we have the following unitary isomorphism between $(L^2)$ and the Boson Fock space $\Gamma(H_{\C})$: 
\begin{equation}
	(L^2) \ni \Phi(x)\ = \sum\limits_{n=0}^\infty \left\langle :x^{\otimes n}:, f_n\right\rangle \ \leftrightarrow \ (f_n) \sim \Phi \in \Gamma(H_{\C}),\ f_n\in L^2(\R,dt)_\C^{\hat\otimes n},
\end{equation} 
where $:x^{\otimes n}:$ denotes the Wick ordering of $x^{\otimes n}$ and $\,^{\hat\otimes n}$ denotes the symmetric tensor product of order $n$. Moreover the $(L^2)$ norm of $\Phi \in (L^2)$ is given by
\begin{equation*}
	\left\|\Phi\right\|_0^2 = \sum\limits_{n=0}^\infty n! \left|f_n\right|_0^2 .
\end{equation*} 
We denote by $\langle \! \langle.,.\rangle \! \rangle$ the canonical $\C$ bilinear form on $(S)_\beta^* \times (S)_\beta$. For each $\Phi \in (S)^*_\beta$ there exists a unique sequence $\left( F_n \right)_{n=0}^\infty, F_n \in (S_{\C}^{\hat{\otimes} n})^*$ such that 
\begin{equation}\label{def_duality}
	\langle \! \langle \Phi, \varphi \rangle \! \rangle = \sum\limits_{n=0}^\infty n! \left\langle F_n,f_n\right\rangle,\indent   (f_n) \sim \varphi \in (S)_\beta .
\end{equation} 
Thus we have, see e.g.~\cite{Hida1, Obata1, Kuo2}: 
$
	(S)_\beta \ni \Phi \sim (f_n), 
$
if and only if for all $ p \in \R$   we have 
$$\left\|\Phi\right\|_{p,\beta} := \left(\sum\limits_{n=0}^\infty (n!)^{1+\beta} \left|f_n\right|_p^2\right)^\frac{1}{2} < \infty.$$
Moreover for its dual space we obtain 
$
	(S)_\beta^* \ni \Phi \sim (F_n), 
$
if and only if there exists a $ p \in \R$  such that $$\left\|\Phi\right\|_{p,-\beta}:= \left(\sum\limits_{n=0}^\infty (n!)^{1-\beta} \left|F_n\right|_p^2\right)^\frac{1}{2} < \infty.$$

\noindent For $p \in \R$ we define 
$$
(S)_{p,\beta}:= \left\{\varphi \in (L^2):\ \left\|\varphi\right\|_{p,\beta} < \infty\right\}
\text{ and }
(S)_{p,-\beta}:= \left\{\varphi \in (S)_\beta^*: \left\|\varphi\right\|_{p,-\beta} < \infty\right\}.
$$
We then obtain
$$
	(S)_\beta:= \mathrm{proj}\!\!\lim_{p\rightarrow\infty} \, \, (S)_{p,\beta}
$$
and 
$$
	(S)_{\beta}^* = \mathrm{ind}\!\!\!\!\lim_{p\rightarrow -\infty}\, \, (S)_{p,-\beta}.
$$

\noindent Moreover $(S)_\beta$ is a nuclear (F)-space. We use the abbreviation $(S):=(S)_0.$\\
\noindent The exponential vector or Wick ordered exponential is defined by 
\begin{equation}
	\Phi_\xi(x) := \sum\limits_{n=0}^\infty \frac{1}{n!} \left\langle :x^{\otimes n}:, {\xi}^{\otimes n}\right\rangle ,
\end{equation}  
where $\xi \in S_{\C}$ and $x \in S^*$.\\
For $y \in S^*_\C$ we use the same notation and define $\Phi_y \in (S)_\beta^*$ by:
\begin{equation*}
	(S)_\beta \ni \psi \sim (f_n)_{n\in\mathbb{N}}:\quad  \langle \! \langle \psi, \Phi_y\rangle \! \rangle := \sum\limits_{n=0}^{\infty} \left\langle y^{\otimes n},f_n\right\rangle.
\end{equation*}
Since $\Phi_\xi \in (S)_{\beta}$, for $\xi \in S_{\mathbb{C}}$ and $0\leq\beta<1$, we can define the so called $S$ transform of $\Psi\in (S)^*_{\beta}$ by
$$ S(\Psi)(\xi) =\langle \! \langle \Psi, \Phi_\xi  \rangle \! \rangle.$$
The $S$ transform can be used to characterize the Hida distributions via a space of ray analytic functions, which is due to the well known characterization theorem, see e.g.~ \cite{Hida1,Obata1,Kuo2,KLWS96+}. 
 
We call $S(\Psi)(0)=\langle \! \langle \Psi, 1\!\!1 \rangle \! \rangle$ the generalized expectation of $\Psi \in (S)^*_{\beta}$.\\
The Wick product of $\Psi_1 \in(S)^*_{\beta} $ and $\Psi_2\in (S)^*_{\beta}$ is defined by $$\Psi_1\diamond\Psi_2 := S^{-1}(S(\Psi_1)\cdot S(\Psi_2)) \in (S)^*_{\beta},$$ see e.g.~\cite{Hida1,Kuo2,Obata1}.\\

Naturally we can define a directional derivative on $(S)$ by
$$\partial_{u} \Phi(x) = \sum_{n=1}^{\infty} n\langle :x^{\otimes n-1}: , \langle f_n, u \rangle \, \rangle,$$
where $\langle f_n, u \rangle$ denotes the contraction of $f_n \in S(\mathbb{R})^{\hat{\otimes} n}_{\mathbb{C}}$ with respect to $u \in S^*.$

It is known that $\partial_{u} \in L((S),(S))$, see e.g.~\cite{Obata1}.

It is shown, see e.g.~\cite{Obata1} that $\partial_{u}$ is indeed a derivation on the space of Hida test functions $(S)$.\\
In Physics applications it plays the role of the annihilation operator in the Fock space, while its dual operator is the creation operator, also known as Skorokhod integral, see e.g.\cite{Hida1, Kuo2, Obata1}.

There are several studies on Poisson algebraic structures on the Hida Test function space see e.g.~\cite{Leandre1, Leandre2, LeandreDito} and their q-deformation. We will follow this streamline here, but exploit the derivation structure of the derivative. 

For this we work on the triple $$(S) \subset L^2(S'(\R,\R^2) \subset (S)^*$$

\begin{thm}
Let $\Phi, \Psi \in (S)$ and $Q\in \mathcal{L}(L^2(\R, \R)) $ a symmetric trace class operator with eigenvalues $(\lambda_n)_{n \in \mathbb{N}}$ and corresponding eigenvectors $e_n \in S(\R)$. We define the Poisson bracket of $\Phi$ and $\Psi$ by
$$
\{\Phi, \Psi \}_Q = \sum_{n=0}^{\infty} \lambda_n(\partial_{q_n} \Phi \partial_{p_n} \Psi - \partial_{p_n} \Phi \partial_{q_n} \Psi),.
$$
where $q_n = (e_n,0)$ and $p_n =(0,e_n).$
With this definition $((S),\{, \})$ is an infinite dimensional Poisson algebra.
\end{thm}

\textbf{Proof:}
For $\Phi, \Psi \in (S)$ we have also the derivative is in $(S)$. However it is a priori unclear if the the infinite series is still a Hida test function. For this we show that indeed the Poisson bracket is in all $(H^p)$ spaces. It is enough to show this for the first part. We have for all $p \geq 0$ and $q>0$:
\begin{eqnarray*}
\| \sum_{n=0}^{\infty} \lambda_n(\partial_{q_n} \Phi \partial_{p_n} \Psi\|_p &\leq&  \sum_{n=0}^{\infty} \| \lambda_n(\partial_{q_n} \Phi \partial_{p_n} \Psi\|_p \\&\leq& C \max_n |e_n|_{-q}  \sum_{n=0}^{\infty} |\lambda_n| \|  \Phi\|_{p+q} \| \Psi\|_{p+q} \\&=& C \max_n |e_n|_{-q} \|  \Phi\|_{p+q} \| \Psi\|_{p+q}\sum_{n=0}^{\infty} |\lambda_n| < \infty.
\end{eqnarray*}
Leibniz rule, bilinearity and Jacobi identity follow directly from the gradient structure and the product rule of the derivative. 
\hfill $\blacksquare$

Define operators $Q,R\in\mathcal{L}((S),\mathcal{O}_{\infty})$ for $\Phi \in (S)$ as follows
\[
Q(\Phi):=\sum\limits_{i,j=1}^{\infty}\langle\! \langle\{\Phi, b_i\}, b_j \rangle\! \rangle S_i  S_j^*,
\]
\[
R(\Phi):=\sum\limits_{i,j=1}^{\infty} \langle\! \langle \Phi \cdot b_j ,  b_i \rangle\! \rangle S_i  S_j^*, \quad \Phi\in (S).
\]
where $\{b_i\}_{i=1}^{\infty} \subset (S)$ is an orthogonal system in $(S)$ extending to $(S)^*$. Moreover $R$ is well defined since $(S)$ is a Banach algebra, see e.g.\cite{Hida1}. 
Then we have
\begin{lem}\label{lem:ComRel}
\begin{equation}
[Q(f),Q(g)]=Q(\{f,g\}),\quad[Q(f),R(g)]=R(\{f,g\}),\label{eqn:Quantizationf_1}
\end{equation}
\begin{equation}
Q(g)R(f)+Q(f)R(g)=Q(fg),\quad R(f)R(g)=R(fg), f,g\in (S) \label{eqn:Leibnitzf_1}
\end{equation}
\end{lem}

\begin{proof}
Follows immediately as before
\end{proof}
\begin{defin}
Define quantization $\widehat{Q}\in \mathcal{L}((S),\mathcal{O}_{\infty})$ as follows
\[
\widehat{Q}:=R-2i Q.
\]
\end{defin}

\begin{thm}\label{thm:QuantInfin}
$\widehat{Q}\in \mathcal{L}((S),\mathcal{O}_{\infty})$ satisfies for $j,k \in \mathbb{N}$
\begin{eqnarray}
\widehat{Q}(1) &= Id,\label{eqn:IdProp}\\ 
{}[\widehat{Q}(f),\widehat{Q}(g)] &= -2i \widehat{Q}(\{f,g\}),\label{eqn:LieBracketProp}\\ 
{}[\widehat{Q}(q_k),\widehat{Q}(q_j)] &= [\widehat{Q}(p_k),\widehat{Q}(p_j)]=0, [\widehat{Q}(q_k),
\widehat{Q}(p_j)] =  -2i \delta_{kj}Id, k,j\in\mathbb{N}.\label{eqn:CCRAnalogueProp}
\end{eqnarray}
Furthermore, if $\phi:\mathbb{R}\to\mathbb{R}$ is an analytic function then
\begin{equation}
\Re \widehat{Q}(\phi(f))=\phi(\Re\widehat{Q}(f))\\\label{eqn:vonNeumannRuleProp}
\end{equation}
\end{thm}
\begin{proof}
Similar to the proof of Theorem \ref{thm:QuantFin}. 
\end{proof}

\end{document}